\documentclass[runningheads,envcountsame,envcountsect]{llncs}

\usepackage[T1]{fontenc}
\usepackage{graphicx}
\usepackage{amsfonts}
\usepackage{latexsym,amssymb,amsmath}
\usepackage{placeins}

\usepackage[ruled]{algorithm}
\usepackage[noend]{algpseudocode}
\algdef{SE}[DOWHILE]{Do}{doWhile}{\algorithmicdo}[1]{\algorithmicwhile\ #1}
\usepackage{comment}

\usepackage{caption}
\usepackage{subcaption}
\usepackage{wrapfig}
\usepackage{multirow}

\usepackage[dvipsnames]{xcolor}
\definecolor{addedblue}{RGB}{38,97,176} %
\definecolor{addedgreen}{RGB}{20,120,50} %

\usepackage{hyperref}

\usepackage[bibliography=common,appendix=inline]{apxproof}
\newtheoremrep{theorem}{Theorem}[section]
\newtheoremrep{lemma}[theorem]{Lemma}

\newcounter{localclaim}
\newcommand{\resetclaims}{\setcounter{localclaim}{0}}
\newcommand{\localclaim}[1]{%
  \refstepcounter{localclaim}\par\smallskip\noindent\textbf{Claim~\arabic{localclaim}.} \emph{#1}\par\smallskip}
\newcommand{\localclaimproof}{\emph{Proof of Claim~\arabic{localclaim}.}}

\title{Search and Rescue on the Plane}
\titlerunning{Search and Rescue on the Plane}

\author{Jared Coleman\inst{1} \and
Evangelos Kranakis\inst{2} \and
Danny Krizanc\inst{3} \and
Oscar Morales-Ponce\inst{4}}

\authorrunning{J. Coleman, E. Kranakis, D. Krizanc, and O. Morales-Ponce}

\institute{Loyola Marymount University, Los Angeles, USA \\ \email{jared.coleman@lmu.edu} \and
Carleton University, Ottawa, Canada \\ \email{kranakis@scs.carleton.ca} \and
Wesleyan University, Middletown, USA \\ \email{dkrizanc@wesleyan.edu} \and
California State University, Long Beach, USA \\ \email{Oscar.MoralesPonce@csulb.edu}}

\begin{document}
\maketitle

\begin{abstract}
We study a planar variant of the search and rescue problem whereby an agent starting at an arbitrary position $P_{\theta,r} = (r\cos\theta, r\sin\theta)$ in the plane must locate an object at an unknown position on the positive $x$-axis and deliver it to the origin. Our main contribution is to characterize the optimal form of any competitive algorithm, derive closed-form expressions for the competitive ratio, and identify a critical angle $\theta^* \approx 15.6^\circ$ which yields a phase transition to optimal competitive search and delivery in the following sense. For each angle $-\pi \leq \theta \leq \pi$ we compute a checkpoint (landing position on the $x$-axis) where the agent must go first prior to initiating a search on the $x$-axis in order to optimize the competitive ratio of search and delivery. We show that if $|\theta| \geq \theta^*$ then the checkpoint is at the origin, while if $|\theta| < \theta^*$ then the agent should land at the checkpoint $(r \cdot k_{|\theta|}, 0)$ on the $x$-axis, where $k_{|\theta|}$ is a real number given by an explicit formula we present. 

\end{abstract}

\FloatBarrier
\section{Introduction}\label{sec:intro}

Search and rescue operations represent a fundamental challenge in mobile robotics and autonomous systems.
While the one-dimensional variant of this problem has been studied, many real-world scenarios involve agents operating in higher-dimensional spaces where geometric considerations become important. In this paper, we investigate a planar search and rescue problem.
We consider an agent starting at an arbitrary position given as $P_{\theta,r} = (r\cos\theta, r\sin\theta)$ in polar coordinates in the plane, at distance $r > 0$ from the origin $O = (0,0)$ and angle $\theta$ from the positive $x$-axis. An object requiring delivery to the origin is located at an unknown position $D = (d, 0)$ on the positive $x$-axis, where $d \in (0, \infty)$. The agent can move with maximum unit speed and can change directions instantaneously. 
\begin{figure}
    \centering
    \includegraphics[width=.75\linewidth]{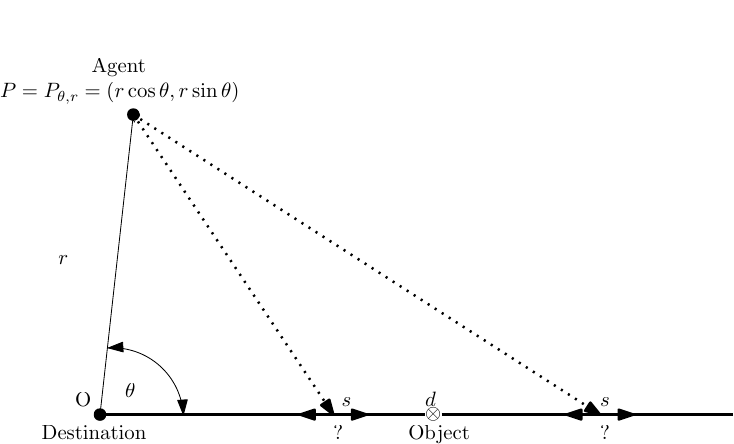}
    \caption{An agent starting at a point $P_{\theta,r}$ is searching for a lost object at an unknown location $(d,0)$. The goal is to design a strategy leading to an ``optimal'' landing position $(s,0)$ on the $x$-axis, so as to find the object and deliver it to the destination at $O$.} 
    \label{fig:example}
\end{figure}
Figure~\ref{fig:example} illustrates the search and delivery problem.
The objective is to give a search algorithm with optimal worst-case performance that finds the object and delivers it to the origin $O$.

As in the one-dimensional case, we measure an algorithm's performance using competitive analysis. An online algorithm must compete against an optimal offline algorithm that knows the object's location ahead of time. The competitive ratio is the worst-case ratio of the online delivery time to the optimal offline delivery time, taken over all possible object locations $d > 0$. A key distinction from the one-dimensional problem is that for each such point $P_{\theta,r}=(r\cos\theta, r\sin\theta)$ the angle $\theta$ introduces geometric complexity into the search strategy. As we show in Section~\ref{sec:model}, the competitive ratio depends only on $\theta$ and not on the starting distance $r$, allowing us to reduce the analysis to the case where the agent lies on the unit circle without loss of generality.

\FloatBarrier
\subsection{Model and Notation}
We consider a single agent operating in the plane. The agent starts at position $P_{\theta,r} = (r\cos\theta, r\sin\theta)$, at distance $r > 0$ from the origin $O = (0,0)$ and angle $\theta$ from the positive $x$-axis. The object requiring delivery is located at an unknown point $D = (d, 0)$ on the positive $x$-axis at distance $d \in (0,\infty)$ from the origin. The destination for delivering the object is the origin $O$. The agent moves at constant unit speed and can change directions instantaneously. The agent can carry the object once found and must physically transport it to the origin; the speed of the agent is not affected when carrying the object.

The distance from the agent's starting position to a point at distance $q \geq 0$ on the positive $x$-axis is
$$\|P_{\theta,r} - (q,0)\| = \sqrt{r^2 - 2qr\cos\theta + q^2}.$$
In the offline setting, the agent knows the value of $d$ and can plan an optimal route. The optimal offline strategy is to move directly from $P_{\theta,r}$ to $D$ and then to $O$, achieving delivery time $\sqrt{r^2 - 2dr\cos\theta + d^2} + d$.

In the online setting, the agent does not know $d$ and must design a search strategy that minimizes the worst-case competitive ratio over all possible values of $d$.
Given an online algorithm $\mathcal{A}$, let $T_{\mathcal{A}}(d, \theta, r)$ denote the delivery time when the object is at distance $d$ and the agent starts at $P_{\theta,r}$. The \emph{competitive ratio} of $\mathcal{A}$ is given as
$$\mathrm{CR}(\mathcal{A}, \theta, r) = \sup_{d > 0} \frac{T_{\mathcal{A}}(d, \theta, r)}{\sqrt{r^2 - 2dr\cos\theta + d^2} + d}.$$
An algorithm is \emph{optimal} if it minimizes $\mathrm{CR}(\mathcal{A}, \theta, r)$ over all online algorithms $\mathcal{A}$ for all values of $r$ and $\theta$.

\FloatBarrier
\subsection{Related Work}\label{sec:related}

Search problems have been a mainstay in the research community because of their numerous applications in distributed computing and other areas of computer science~\cite{alpern2003theory}. The cow-path problem and its variants~\cite{baezayates1993searching,beck1964linear,gal2010search} provide foundational results for search problems, though these focus primarily on the search phase without subsequent delivery requirements. The one-dimensional problem has been thoroughly analyzed, with optimal algorithms known for both single and multiple agents. Evacuation and group search problems~\cite{czyzowicz2019group,czyzowicz2021group} consider multiple agents but typically in symmetric environments.  %

The unconstrained version of planar search (locating an unknown point in the plane starting from a fixed origin) has been studied extensively. Gal in \cite{gal1980} proposed a spiral search algorithm which was proved to be optimal by Langetepe in~\cite{langetepe2010spiral} by showing that logarithmic spirals achieve the optimal competitive ratio of $17.289\ldots$ in this setting. Collaborative search with multiple speed robots can be found in \cite{georgiou2025spirals}.
When partial angular information about the target is available, costs can drop sharply: Bouchard et al.~\cite{bouchard2020treasure} show that repeated angular hints of aperture at most $\pi$ reduce the cost of treasure hunt from $\Theta(D^2)$ to $\Theta(D)$, where $D$ is the initial distance of the treasure from the agent. Our setting lies at the extreme of this spectrum (the target's direction from the origin is fully known, only its distance is unknown), which admits an exact competitive-ratio analysis. In another related line of work Eubeler et al.~\cite{eubeler_et_al:DagSemProc.06421.5} analyze the competitive ratio of locating the source of a target ray when the searcher must commit to a search trajectory in the plane.

There are several studies on delivery for multirobot systems. Carvalho et al.~\cite{carvalho2019efficient,carvalho2021fast} study the problem in weighted graphs. They provide an offline algorithm that runs in $O(k n \log n + km)$ time where $k$ is the number of robots, $n$ is the number of nodes, and $m$ the number of edges of the graph.
Delivery with energy constraints is considered in~\cite{chalopin2013data,chalopin2014data}. Message delivery in the plane~\cite{coleman2021message} differs from object delivery considered in our work in that in the former messages can be replicated. In addition, there is extensive literature on pickup and delivery, e.g., ~\cite{berbeglia2010dynamic} which considers dynamic problems whereby objects and people have to be delivered with real time constraints, but the approach is not in the spirit of our paper. The authors in~\cite{coleman2024optimaldeliveryfaultydrone} study the competitive ratio of object delivery in the presence of crash-faulty agents.

Search and rescue has been studied on a disk (under the name search and fetch) for agents starting at its center while the object and destination are placed at the perimeter~\cite{georgiou2016search_one}.
The present work on search and rescue extends the framework introduced for the line~\cite{coleman_2023_sar} to planar environments. The planar variant analyzed here focuses on competitive analysis and the design and solution provided (as we will see) introduces subtle geometric considerations that fundamentally alter the structure of the one-dimensional search and rescue problem.

\subsection{Main Contributions}\label{sec:contributions}

Our main result determines the optimal strategy for an agent starting at any position $P_{\theta,r} = (r\cos\theta, r\sin\theta)$, where $-\pi \leq \theta \leq \pi$. Details of the pseudo-code leading to the optimal strategy are given in Algorithm~\ref{alg:optimal_general}.

\begin{algorithm}[t]
\caption{Optimal algorithm for agent at $P_{\theta,r} = (r\cos\theta, r\sin\theta)$}\label{alg:optimal_general}
\begin{algorithmic}[1]
\Require Agent at $P_{\theta,r}$; object at unknown $(d,0)$; delivery destination $O = (0,0)$
\If{$|\theta| \leq \arccos\!\left(\frac{1}{12}\left(-8 + \sqrt[3]{1864 - 312\sqrt{33}} + 2\sqrt[3]{233 + 39\sqrt{33}}\right)\right)$}
    \State $k_{|\theta|} \gets \frac{5 + 2\cos\theta + \cos 2\theta + \sqrt{2\cos^2\!\theta\,(11 + 4\cos\theta + \cos 2\theta)}}{8}$
    \State Move directly from $P_{\theta,r}$ to the checkpoint $(r \cdot k_{|\theta|},\, 0)$
    \State Move from $(r \cdot k_{|\theta|},\, 0)$ to the origin $O$
\Else
    \State Move directly from $P_{\theta,r}$ to the origin $O$
\EndIf
\State Move along the positive $x$-axis from $O$ until the object is found at $D = (d, 0)$
\State Deliver the object from $D$ to the origin $O$
\end{algorithmic}
\end{algorithm}

\begin{theorem}\label{thm:main_general}
Algorithm~\ref{alg:optimal_general} has optimal competitive ratio.
\end{theorem}

The proof of Theorem~\ref{thm:main_general} follows from the results established in Sections~\ref{sec:structure}--\ref{sec:analysis}. Our specific contributions are:
\begin{enumerate}
\item We show that the competitive ratio is invariant under scaling of $r$, reducing the problem to the unit circle (Lemma~\ref{lem:reduction}).
\item We prove that any optimal algorithm has a simple canonical form $A_s$ (Algorithm~\ref{alg:as}): move directly to a point at distance $s$ on the $x$-axis, return to the origin, then search away from the origin (Theorem~\ref{thm:algorithm_form}).
\item We derive closed-form expressions for the competitive ratio of $A_s$ as a function of $s$ and $\theta$ (Lemmas~\ref{lem:case_d_leq_s} and~\ref{lem:case_d_gt_s}).
\item We determine the optimal parameter $s^*(\theta) = k_{|\theta|}$ in closed form and identify the critical angle $\theta^*$ (Theorem~\ref{thm:critical_angle}): for $|\theta| < \theta^*$ the agent should go to $r \cdot k_{|\theta|}$ first, while for $|\theta| \geq \theta^*$ the agent should go directly to the origin.
\end{enumerate}

\FloatBarrier
\section{Preliminaries}\label{sec:model}

In this section we introduce notation and preliminary results which are important in understanding the rest of the paper. 

\subsection{Reduction to the Unit Circle}\label{sec:reduction}

We may assume without loss of generality that $\theta \in [0, \pi)$, since the problem is symmetric about the $x$-axis: an agent starting at angle $-\theta$ (i.e., with $y < 0$) faces an identical problem by reflection.

\begin{lemma}\label{lem:reduction}
For any $r > 0$, $\mathrm{CR}(\mathcal{A}, \theta, r) = \mathrm{CR}(\mathcal{A}', \theta, 1)$, where $\mathcal{A}'$ is the algorithm $\mathcal{A}$ expressed in coordinates rescaled by $r$. In particular, the competitive ratio is independent of $r$.
\end{lemma}

\begin{proof}
Rescale coordinates by setting the unit of length to $r$. In the rescaled system, the agent starts at unit distance from the origin, the object is at distance $d/r$ on the positive $x$-axis, and all travel distances are divided by $r$. Both the online time $T_{\mathcal{A}}$ and the offline time scale by the same factor of $r$, so their ratio is unchanged.

\end{proof}

By Lemma~\ref{lem:reduction}, we henceforth set $r = 1$ and write $P_\theta = P_{\theta,1} = (\cos\theta, \sin\theta)$. For convenience, we define $\rho(\theta, q)$ as the distance from $P_\theta$ to a point on the positive $x$-axis at (Euclidean) distance $q$ from the origin:
\begin{equation}\label{eq:rho_def}
\rho(\theta, q) = \sqrt{1 - 2q\cos\theta + q^2}.
\end{equation}
The optimal offline delivery time simplifies to $T^*(d, \theta) = \rho(\theta, d) + d$, and the competitive ratio of an algorithm $\mathcal{A}$ becomes
$$\mathrm{CR}(\mathcal{A}, \theta) = \sup_{d > 0} \frac{T_{\mathcal{A}}(d, \theta)}{T^*(d, \theta)}.$$

\FloatBarrier
\subsection{The Optimal Algorithm}\label{sec:structure}

We begin by characterizing the form that any optimal competitive algorithm must take. We define a family $(A_s: s \geq 0)$ of algorithms, where $A_s$ is Algorithm~\ref{alg:as}: the agent moves directly from $P_\theta$ to the point $S = (s, 0)$ on the $x$-axis, then to the origin $O$, and finally searches along the positive $x$-axis away from the origin until encountering the object at $D$, delivering it to $O$.

\begin{algorithm}[t]
\caption{Algorithm $A_s$ for the agent starting at $P_\theta = (\cos\theta, \sin\theta)$}\label{alg:as}
\begin{algorithmic}[1]
\Require Checkpoint distance $s \geq 0$; agent starts at $P_\theta$; object at unknown $(d,0)$
\State Move directly from $P_\theta$ to the checkpoint $S = (s, 0)$ \Comment{cost: $\rho(\theta, s)$}
\State Move from $S$ to the origin $O = (0, 0)$ \Comment{cost: $s$}
\State Move away from $O$ until the object is found at $D = (d, 0)$ \Comment{cost: $d$}
\State Deliver the object from $D$ to the origin $O$ \Comment{cost: $d$}
\end{algorithmic}
\end{algorithm}

\begin{theorem}\label{thm:algorithm_form}
There exists an $s \geq 0$ for which $A_s$ is optimal.
\end{theorem}

\begin{proof}
\resetclaims
We establish this result through a series of claims concerning the structure of an optimal algorithm for our search and delivery problem.

\localclaim{Optimal algorithms reach the $x$-axis directly on a straight line path.}
\localclaimproof\ Both the object and the destination lie on the $x$-axis, so any path from $P_\theta$ to the $x$-axis longer than the straight line strictly increases travel time without providing information about the object's location.

\localclaim{Once on the $x$-axis, optimal algorithms remain on it.}
\localclaimproof\ Any excursion off the axis strictly increases travel time, since the object and destination both lie on the $x$-axis.

\localclaim{The agent should visit the origin before searching further from it.}
\localclaimproof\
Let $S = (s,0)$ denote the first point where the agent reaches the $x$-axis, at distance $s$ from the origin, and let $S_{\max} = (s_{\max}, 0)$ denote the farthest point from the origin visited before reaching $O$, at distance $s_{\max} \geq s$ from the origin. These two quantities summarize the agent's behavior before its first visit to $O$, no matter how it moves on the axis. The explored region is always the interval $[0, s_{\max}]$, so any reversal merely revisits already-explored points and reveals the object no sooner. Consequently the path costs used below are the cheapest that reach $O$ after exploring out to $s_{\max}$, and the resulting case bounds lower-bound the competitive ratio of \emph{every} online algorithm, not just those that move monotonically from $S$ to $S_{\max}$ to $O$. We analyze the competitive ratio in three cases based on the object's (unknown) location $(d,0)$.

\textbf{Case 1:} If $d \in [s, s_{\max}]$, the object is found before reaching the origin. The competitive ratio is:
$$
\sup_{d \in [s,s_{\max}]} \frac{\rho(\theta, s) + (d - s) + d}{\rho(\theta, d) + d} = \sup_{d \in [s,s_{\max}]} \frac{\rho(\theta, s) + 2d - s}{\rho(\theta, d) + d}.
$$

\textbf{Case 2:} If $d < s$, the object lies closer to the origin than the first point on the axis. The worst-case competitive ratio is:
$$
\sup_{d \in (0,s)} \frac{\rho(\theta, s) + (s_{\max} - s) + s_{\max}}{\rho(\theta, d) + d} = \rho(\theta, s) + 2s_{\max} - s.
$$
The above expression follows since the left-hand side is decreasing in $d$ for $d < s$, so the supremum is achieved as $d \to 0^+$.

\textbf{Case 3:} If $d > s_{\max}$, the object lies beyond the farthest explored point. The competitive ratio is:
$$
\sup_{d > s_{\max}} \frac{\rho(\theta, s) + (s_{\max} - s) + s_{\max} + 2d}{\rho(\theta, d) + d} = \sup_{d > s_{\max}} \frac{\rho(\theta, s) + 2s_{\max} - s + 2d}{\rho(\theta, d) + d}.
$$

The overall competitive ratio is the maximum of the three cases above, hence at least $\max\{\text{Case 2},\, \text{Case 3}\}$.
Observe that both Cases~2 and~3 are strictly increasing with respect to $s_{\max}$. Since $s_{\max} \geq s$, this lower bound is minimized when $s_{\max} = s$. At $s_{\max} = s$ the domain of Case~1 collapses to the single point $d = s$, giving ratio $1$, so Case~1 is not the binding constraint.

It follows from the Claims above that an optimal algorithm involves the agent moving directly from $P_\theta$ to $S$, then to the origin, and finally searching along the positive $x$-axis away from the origin.

\end{proof}

\FloatBarrier
\section{Analysis of Algorithm Family $A_s$}\label{sec:analysis}

Having characterized the optimal algorithm form, we now analyze the competitive ratio for the algorithm family $(A_s: s\geq 0)$. Throughout this section we assume $\theta \in (0, \pi)$ and $s \geq 0$, and write $\alpha = \rho(\theta, s) + s$ and $\gamma = \cos\theta$; when the dependence on $s$ and $\theta$ matters, we write quantities like $d_-$ (introduced in Lemma~\ref{lem:case_d_gt_s}) as $d_-(s, \theta)$. The boundary case $\theta = 0$ reduces to search and rescue on a ray~\cite{coleman_2023_sar} and is recovered as a limit in Theorem~\ref{thm:critical_angle}.

\subsection{Competitive Ratio of $A_s$}\label{sec:cr_as}

Algorithm $A_s$ moves from $P_\theta$ to $S$, then to $O$, and finally searches along the positive $x$-axis away from the origin. The competitive ratio depends on whether the object is encountered during the initial traversal from $S$ to $O$ or requires additional search beyond $O$. 

We denote by $\mathrm{CR}_{d \leq s}(s, \theta)$ the worst-case competitive ratio of $A_s$ restricted to object positions $d \in (0, s]$.

\begin{lemma}\label{lem:case_d_leq_s}
For $s > 0$, $\mathrm{CR}_{d \leq s}(s, \theta) = \rho(\theta, s) + s.$
\end{lemma}

\begin{proof}
For $d \in (0, s]$ the agent encounters the object on the way from $S$ to $O$, so the online cost is $\rho(\theta,s) + s$ and the competitive ratio is
$$
\frac{\rho(\theta, s) + s}{\rho(\theta, d) + d}.
$$
Since $\rho(\theta, d) + d$ is increasing in $d$, the supremum is achieved as $d \to 0^+$, giving $\rho(\theta,s) + s$ (using $\rho(\theta, 0) = 1$).

\end{proof}

For $d > s$, the object lies beyond $S$. The agent traverses $P_\theta \to S \to O$ (in this order) without encountering the object, then searches along the positive $x$-axis to distance $d$, and returns to $O$. The competitive ratio is
\begin{equation}\label{eq:cr_d_gt_s}
\mathrm{CR}(s, d, \theta) = \frac{\rho(\theta, s) + s + 2d}{\rho(\theta, d) + d}.
\end{equation}
We denote by $\mathrm{CR}_{d > s}(s, \theta) = \sup_{d > s} \mathrm{CR}(s, d, \theta)$ the worst-case competitive ratio of $A_s$ restricted to object positions $d > s$.

\begin{lemma}\label{lem:case_d_gt_s}
Let
$$
d_- = \frac{\alpha(1+\gamma^2) + 2\gamma - \alpha\sqrt{(1-\gamma^2)(1+\alpha\gamma)}}{2\gamma(\alpha + \gamma)} .
$$
where $\alpha = \rho(\theta, s) + s$, $\gamma = \cos\theta$. Then,
$$
\mathrm{CR}_{d > s}(s, \theta) = \begin{cases} \mathrm{CR}(s, d_-, \theta) & \text{if } d_- \text{ is real and } d_- > s, \\ \mathrm{CR}(s, s, \theta) & \text{otherwise}. \end{cases}
$$
\end{lemma}

\begin{proof}
From Equation~\eqref{eq:cr_d_gt_s}, $\mathrm{CR}(s,d,\theta) = (\alpha + 2d)/(\rho(\theta,d) + d)$ where $\alpha = \rho(\theta,s) + s$.
By the quotient rule,
$$
\frac{\partial}{\partial d}\mathrm{CR}(s,d,\theta) = \frac{2(\rho(\theta,d) + d) - (\alpha + 2d)\!\left(\frac{d - \gamma}{\rho(\theta,d)} + 1\right)}{(\rho(\theta,d) + d)^2},
$$
where the derivative of the denominator uses $\frac{\partial}{\partial d}\rho(\theta,d) = \frac{d - \gamma}{\rho(\theta,d)}$, which follows from $\rho(\theta,d) = \sqrt{1 - 2d\gamma + d^2}$ with $\gamma = \cos\theta$.

Setting the numerator to zero and multiplying through by $\rho(\theta,d)$:
\begin{align}
2\rho(\theta,d)(\rho(\theta,d) + d) &= (\alpha + 2d)(\rho(\theta,d) + d - \gamma) \notag \\
2\rho^2 + 2d\rho &= (\alpha + 2d)\rho + (\alpha+2d)(d - \gamma) && \text{(expanding)} \notag \\
-\alpha\,\rho &= (\alpha + 2d)(d - \gamma) - 2\rho^2 && \text{(collecting $\rho$ terms)} \notag \\
 &= \alpha d + 2d\gamma - \alpha\gamma - 2 && \text{(using $\rho^2 = 1 - 2d\gamma + d^2$)} \notag \\
\alpha\,\rho(\theta, d) &= 2 + \alpha\gamma - d(\alpha + 2\gamma). \label{eq:cr_critical}
\end{align}
Since the left-hand side is positive, so must be the right; when $\alpha + 2\gamma > 0$ this imposes $d < (2 + \alpha\gamma)/(\alpha + 2\gamma)$.
Squaring~\eqref{eq:cr_critical} and substituting $\rho^2 = 1 - 2d\gamma + d^2$ yields a quadratic $Ad^2 + Bd + C = 0$ with
\begin{align*}
A &= \alpha^2 - (\alpha+2\gamma)^2 = -4\gamma(\alpha + \gamma), \\
B &= -2\alpha^2\gamma + 2(2+\alpha\gamma)(\alpha+2\gamma) = 4\bigl(\alpha(1 + \gamma^2) + 2\gamma\bigr), \\
C &= \alpha^2 - (2+\alpha\gamma)^2 = -(4 - \alpha^2(1 - \gamma^2)),
\end{align*}
and discriminant $\Delta = B^2 - 4AC = 16\alpha^2(1 - \gamma^2)(1 + \alpha\gamma)$.
When $\Delta < 0$, the roots are complex and $\mathrm{CR}(s,d,\theta)$ (as defined in Equation~\eqref{eq:cr_d_gt_s}) has no critical points in $d$; since $\lim_{d \to \infty}\mathrm{CR}(s,d,\theta) = 1$, the function is strictly decreasing on $(0,\infty)$. When $\theta \in (\pi/2, \pi)$ (i.e., $\gamma < 0$), $\Delta < 0$ occurs when $\alpha > 1/|\gamma|$.

When $\Delta \geq 0$, the quadratic formula gives two roots:
$$
d_{\pm} = \frac{-B \pm \sqrt{\Delta}}{2A} = \frac{\alpha(1+\gamma^2) + 2\gamma \pm \alpha\sqrt{(1-\gamma^2)(1+\alpha\gamma)}}{2\gamma(\alpha + \gamma)}.
$$

Squaring may introduce a spurious root that does not satisfy the original (unsquared) equation~\eqref{eq:cr_critical}. We verify that $d_-$ is the unique valid critical point by checking the positivity constraint $2 + \alpha\gamma - d(\alpha + 2\gamma) > 0$ from~\eqref{eq:cr_critical}.

\emph{Case $\gamma > 0$ (i.e., $\theta \in (0, \pi/2)$):}
Here $A = -4\gamma(\alpha + \gamma) < 0$, so the parabola $Ad^2 + Bd + C$ opens downward. Since $B > 0$ and $A < 0$, we have $d_- < d_+$, and both roots are positive.
The right-hand side of~\eqref{eq:cr_critical} is a decreasing linear function of $d$ (since $\alpha + 2\gamma > 0$) that is positive at $d = 0$. The bound from the positivity constraint is $d < (2 + \alpha\gamma)/(\alpha + 2\gamma)$. One can verify that $d_-$ satisfies this bound while $d_+$ does not, by checking that $d_+ > (2+\alpha\gamma)/(\alpha+2\gamma) > d_-$ (using the explicit formulas). Hence $d_+$ is spurious and $d_-$ is the unique valid critical point.

\emph{Case $\gamma < 0$ (i.e., $\theta \in (\pi/2, \pi)$) with $\Delta \geq 0$:}
Here $\alpha + \gamma > 0$ (since $\alpha = \rho(\theta,s) + s \geq 1$ and $|\gamma| < 1$), so the denominator $2\gamma(\alpha + \gamma) < 0$. The numerator of $d_+$ is $\alpha(1+\gamma^2) + 2\gamma + \alpha\sqrt{(1-\gamma^2)(1+\alpha\gamma)}$. Since $1 + \gamma^2 > 2|\gamma|$ for $|\gamma| < 1$, we have $\alpha(1+\gamma^2) + 2\gamma > \alpha \cdot 2|\gamma| + 2\gamma = 0$ (as $\gamma < 0$), so the numerator of $d_+$ is positive. With a negative denominator, $d_+ < 0$, so $d_+$ is not a valid root.
For $d_-$, the numerator is $\alpha(1+\gamma^2) + 2\gamma - \alpha\sqrt{(1-\gamma^2)(1+\alpha\gamma)}$, which can be negative (since $\gamma < 0$ reduces the first terms), giving $d_- > 0$ with the negative denominator. Hence $d_-$ is the unique positive root.

When $d_-$ is real and positive, it is the unique critical point of $\mathrm{CR}(s,d,\theta)$ for $d > 0$. Since $\mathrm{CR} \to 1$ as $d \to \infty$, this critical point must be a global maximum.
For the restriction to $d > s$: if $d_- > s$, the maximum over $(s,\infty)$ is attained at $d_-$. If $d_- \leq s$ (or $d_-$ does not exist as a positive real), the function is decreasing on $(s, \infty)$ and the supremum is $\lim_{d \to s^+}\mathrm{CR}(s,d,\theta) = \mathrm{CR}(s,s,\theta)$.

\end{proof}

\begin{lemma}\label{lem:d_minus_decreasing}
$d_-(s, \theta)$ is strictly decreasing in $s$ wherever it is positive.
\end{lemma}

\begin{proof}
Recall that we let $\alpha = \rho(\theta,s) + s$ and $\gamma = \cos\theta$. By the chain rule,
$$
\frac{\partial d_-}{\partial s} = \frac{\partial d_-}{\partial \alpha} \cdot \frac{\partial \alpha}{\partial s}.
$$
For the second factor, $\frac{\partial \alpha}{\partial s} = 1 + \frac{s - \gamma}{\rho(\theta,s)} = \frac{\rho(\theta,s) + s - \gamma}{\rho(\theta,s)}$.
Since $\sin\theta \neq 0$,
$$\rho(\theta,s) = \sqrt{(\gamma - s)^2 + \sin^2\theta} > |\gamma - s|,$$
so $\rho(\theta,s) + s - \gamma > 0$ and $\frac{\partial \alpha}{\partial s} > 0$.
For the first factor, differentiating the expression for $d_-$ with respect to $\alpha$ yields
$$
\frac{\partial d_-}{\partial \alpha} = -\frac{\sin\theta\bigl(\alpha^2 + 3\alpha\gamma + 2 + 2\sin\theta\sqrt{1 + \alpha\gamma}\bigr)}{4(\alpha + \gamma)^2\sqrt{1 + \alpha\gamma}}.
$$
Since $d_-$ is positive, $1 + \alpha\gamma > 0$ by Lemma~\ref{lem:case_d_gt_s}. Also $\alpha + \gamma > 0$: at $s = 0$, $\alpha + \gamma = 1 + \gamma > 0$ since $\gamma \in (-1,1)$, and $\alpha$ is increasing in $s$ as shown above. Therefore $\alpha^2 + 3\alpha\gamma + 2 = \alpha(\alpha + \gamma) + 2(1 + \alpha\gamma) > 0$, every factor is positive, and $\frac{\partial d_-}{\partial \alpha} < 0$.
Thus, $\frac{\partial d_-}{\partial s} < 0$.

\end{proof}

\begin{theoremrep}\label{thm:as_competitive_ratio}
The competitive ratio of Algorithm $A_s$ is:
$$
\mathrm{CR}(A_s, \theta) = \max \{ \mathrm{CR}_{d \leq s}(s, \theta),~ \mathrm{CR}_{d > s}(s, \theta) \} .
$$
\end{theoremrep}

\begin{proof}
Follows directly from Lemmas~\ref{lem:case_d_leq_s} and~\ref{lem:case_d_gt_s}.

\end{proof}

\subsection{Optimal Choice of $s$}\label{sec:optimal_s}

Our goal is to choose $s \geq 0$ to minimize $\mathrm{CR}(A_s, \theta)$. We establish the structure of $\mathrm{CR}(A_s,\theta)$ as a function of~$s$ through a series of lemmas (see Figure~\ref{fig:optimal_s}).

\begin{figure}[t]
\centering
\includegraphics[width=0.75\textwidth]{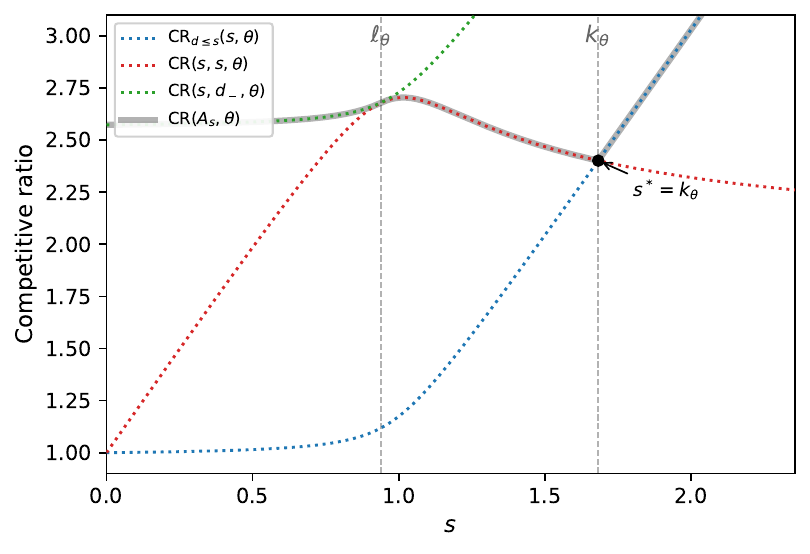}
\caption{Competitive ratio curves as functions of $s$ for $\theta = 10^\circ$. The thick gray curve shows $\mathrm{CR}(A_s,\theta)$: it follows $\mathrm{CR}(s,d_-,\theta)$ for $s \leq \ell_\theta$, then $\mathrm{CR}(s,s,\theta)$ for $\ell_\theta \leq s \leq k_\theta$, then $\mathrm{CR}_{d \leq s}(s,\theta)$ for $s \geq k_\theta$. The minimum is attained at $s^* = k_\theta$.}
\label{fig:optimal_s}
\end{figure}

\begin{lemmarep}\label{lem:key_ineq}
$\mathrm{CR}(s, d_-(s,\theta), \theta) \geq \mathrm{CR}_{d \leq s}(s,\theta).$
\end{lemmarep}

\begin{proof}
By Lemma~\ref{lem:case_d_gt_s}, $d_-(s,\theta)$ is the global maximizer of $\mathrm{CR}(s,d,\theta)$ over $d > 0$ when it exists as a positive real. Since $\lim_{d \to 0^+} \mathrm{CR}(s,d,\theta) = \rho(\theta,s)+s = \mathrm{CR}_{d \leq s}(s,\theta)$, the maximum value at $d_-$ is at least this limit.

\end{proof}

\begin{lemmarep}\label{lem:cr_dminus_increasing}
$\mathrm{CR}(s,d_-(s,\theta),\theta)$ is increasing in~$s$.
\end{lemmarep}

\begin{proof}
For every fixed $d > 0$, $\mathrm{CR}(s,d,\theta) = (\rho(\theta,s)+s+2d)/(\rho(\theta,d)+d)$ is increasing in~$s$, since the numerator increases with $\alpha = \rho(\theta,s)+s$ while the denominator is independent of~$s$. Taking the supremum over $d > 0$ preserves this monotonicity.

\end{proof}

\begin{lemmarep}\label{lem:crossing}
For $\theta \in (0,\pi/2)$, the equation $\mathrm{CR}_{d \leq s}(s,\theta) = \mathrm{CR}(s,s,\theta)$ has exactly two solutions for $s \geq 0$: $s = 0$ and $s = k_\theta$, where
\begin{equation}\label{eq:k_theta}
k_\theta := \frac{5 + 2\cos\theta + \cos 2\theta + \sqrt{2\cos^2\theta\,(11 + 4\cos\theta + \cos 2\theta)}}{8}.
\end{equation}
Moreover,
\begin{align*}
\mathrm{CR}(s,s,\theta) &> \mathrm{CR}_{d \leq s}(s,\theta) \;\text{ for } s \in (0, k_\theta), \\
\mathrm{CR}(s,s,\theta) &< \mathrm{CR}_{d \leq s}(s,\theta) \;\text{ for } s > k_\theta.
\end{align*}
\end{lemmarep}

\begin{proof}
By Lemma~\ref{lem:case_d_leq_s}, $\mathrm{CR}_{d \leq s}(s,\theta) = \rho(\theta,s) + s$, and from~\eqref{eq:cr_d_gt_s}, $\mathrm{CR}(s,s,\theta) = (\rho(\theta,s) + 3s)/(\rho(\theta,s) + s)$. Define
$$
f(s) \;=\; (\rho(\theta,s)+s)^2 - (\rho(\theta,s)+3s),
$$
so that $f(s) = \mathrm{CR}_{d \leq s}(s,\theta)\bigl[\mathrm{CR}_{d \leq s}(s,\theta) - \mathrm{CR}(s,s,\theta)\bigr]$. Since $\mathrm{CR}_{d \leq s} > 0$, $f(s) = 0$ iff $\mathrm{CR}_{d \leq s}(s,\theta) = \mathrm{CR}(s,s,\theta)$.

Writing $\rho = \rho(\theta,s)$ for brevity and $\gamma = \cos\theta$, we reduce $f(s) = 0$ to a polynomial as follows:
\begin{align}
(\rho + s)^2 &= \rho + 3s \notag \\
\rho(2s - 1) &= 3s - s^2 - \rho^2 & &\text{(rearranging)} \notag \\
 &= -1 + 3s + 2s\gamma - 2s^2 & &\text{(substituting $\rho^2$)} \label{eq:crossing_unsquared} \\
(2s-1)^2\rho^2 &= (-1 + 3s + 2s\gamma - 2s^2)^2 & &\text{(squaring both sides)} \notag \\
(2s-1)^2(1 - 2s\gamma + s^2) &= (-1 + 3s + 2s\gamma - 2s^2)^2. \label{eq:squared_crossing}
\end{align}
Expanding both sides of~\eqref{eq:squared_crossing} and subtracting, all terms of degree~4 cancel, yielding the cubic
\begin{equation}\label{eq:cubic_crossing}
2s(1 + \gamma) - s^2(8 + 4\gamma + 4\gamma^2) + 8s^3 = 0.
\end{equation}

Equation~\eqref{eq:cubic_crossing} factors as $s \cdot Q(s) = 0$ where
\begin{equation}\label{eq:quadratic_crossing}
Q(s) = 8s^2 - 2(5 + 2\gamma + \cos 2\theta)\,s + 2(1 + \gamma) = 0.
\end{equation}
The root $s = 0$ is the trivial solution. The quadratic~\eqref{eq:quadratic_crossing} has discriminant
$$
4(5 + 2\gamma + \cos 2\theta)^2 - 64(1+\gamma) = 8\gamma^2(11 + 4\gamma + \cos 2\theta),
$$
yielding roots
$$
s_\pm = \frac{5 + 2\gamma + \cos 2\theta \pm \sqrt{2\gamma^2\,(11 + 4\gamma + \cos 2\theta)}}{8}.
$$

Both $s_+$ and $s_-$ satisfy the squared equation~\eqref{eq:squared_crossing}, but not necessarily the original equation $f(s) = 0$, since squaring may introduce spurious roots.
Recall that $f(s) = 0$ requires $\rho(2s-1) = -1 + 3s + 2s\gamma - 2s^2$. At $s = s_+$, direct substitution confirms this identity holds. At $s = s_-$, the two sides of the unsquared equation have opposite signs, so $f(s_-) \neq 0$ and $s_-$ is extraneous.
Thus the only solutions of $f(s) = 0$ for $s \geq 0$ are $s = 0$ and $s = k_\theta := s_+$.

Observe that $f(s) < 0$ if and only if $\mathrm{CR}(s,s,\theta) > \mathrm{CR}_{d \leq s}(s,\theta)$. Since $f$ is continuous and its only zeros for $s \geq 0$ are $s = 0$ and $s = k_\theta$, the sign of~$f$ is constant on each of the intervals $(0, k_\theta)$ and $(k_\theta, \infty)$. We determine the sign on $(0, k_\theta)$ by computing $f'(0)$: since $\rho(\theta,0) = 1$ and $\rho'(\theta,0) = -\gamma$, we get $f'(0) = 2(1)(-\gamma + 1) - (-\gamma + 3) = -1 - \gamma < 0$, so $f < 0$ on $(0,k_\theta)$. For $(k_\theta, \infty)$, $f(s) \to +\infty$ as $s \to \infty$ (since the $(\rho+s)^2$ term dominates), giving $f > 0$.

\end{proof}

We denote by $k_\theta$ the \emph{optimal checkpoint distance} for angle $\theta$, defined by
\begin{equation}\label{eq:k_theta_def}
k_\theta \;=\; \frac{5 + 2\cos\theta + \cos 2\theta + \sqrt{2\cos^2\!\theta\,(11 + 4\cos\theta + \cos 2\theta)}}{8}.
\end{equation}
By Lemma~\ref{lem:crossing}, $k_\theta$ is the unique positive value of $s$ at which $\mathrm{CR}_{d \leq s}(s,\theta) = \mathrm{CR}(s,s,\theta)$. The next theorem shows that the minimum of $\mathrm{CR}(A_s,\theta)$ over $s \geq 0$ is attained at either $s = 0$ or $s = k_\theta$.

\begin{theoremrep}\label{thm:optimal_s}
$$
\min_{s \geq 0}\, \mathrm{CR}(A_s,\theta) = \min\bigl\{\mathrm{CR}(A_0,\theta),\; \mathrm{CR}(A_{k_\theta},\theta)\bigr\}.
$$
\end{theoremrep}

\begin{proof}
Since $d_-(0,\theta) > 0$ and $d_-(s,\theta)$ is strictly decreasing in~$s$ (Lemma~\ref{lem:d_minus_decreasing}), there exists a unique $\ell_\theta > 0$ with $d_-(\ell_\theta,\theta) = \ell_\theta$. By Lemma~\ref{lem:crossing}, $\mathrm{CR}(\ell_\theta,\ell_\theta,\theta) > \mathrm{CR}_{d \leq \ell_\theta}(\ell_\theta,\theta)$, so $\ell_\theta \in (0, k_\theta)$. We analyze $\mathrm{CR}(A_s,\theta) = \max\{\mathrm{CR}_{d \leq s},\, \mathrm{CR}_{d > s}\}$ on the three intervals $[0, \ell_\theta]$, $[\ell_\theta, k_\theta]$, and $[k_\theta, \infty)$ depicted in Figure~\ref{fig:optimal_s}.

\emph{For $s \in [0, \ell_\theta]$}: $d_-(s,\theta) > s$, so $\mathrm{CR}_{d > s}(s,\theta) = \mathrm{CR}(s,d_-,\theta) \geq \mathrm{CR}_{d \leq s}(s,\theta)$ by Lemma~\ref{lem:key_ineq}. Thus $\mathrm{CR}(A_s,\theta) = \mathrm{CR}(s,d_-,\theta)$, which is increasing by Lemma~\ref{lem:cr_dminus_increasing}. The minimum in this range is $\mathrm{CR}(A_0,\theta)$.

\emph{For $s \in [\ell_\theta, k_\theta]$}: $d_-(s,\theta) \leq s$, so $\mathrm{CR}_{d > s}(s,\theta) = \mathrm{CR}(s,s,\theta) > \mathrm{CR}_{d \leq s}(s,\theta)$ by Lemma~\ref{lem:crossing}. Hence $\mathrm{CR}(A_s,\theta) = \mathrm{CR}(s,s,\theta)$. Now $\mathrm{CR}(s,s,\theta) = 1 + 2s/(\rho(\theta,s)+s)$ has derivative $2(1-s\cos\theta)/(\rho\,(\rho+s)^2)$, which changes sign at most once, so $\mathrm{CR}(s,s,\theta)$ is unimodal. Its minimum on $[\ell_\theta, k_\theta]$ is therefore at an endpoint: either $\mathrm{CR}(\ell_\theta,\ell_\theta,\theta)$, which exceeds $\mathrm{CR}(A_0,\theta)$ by Lemma~\ref{lem:cr_dminus_increasing} since $\ell_\theta > 0$; or $\mathrm{CR}(k_\theta,k_\theta,\theta) = \mathrm{CR}(A_{k_\theta},\theta)$.

\emph{For $s \geq k_\theta$}: $\mathrm{CR}_{d \leq s}(s,\theta) \geq \mathrm{CR}(s,s,\theta) \geq \mathrm{CR}_{d > s}(s,\theta)$ by Lemma~\ref{lem:crossing}, so $\mathrm{CR}(A_s,\theta) = \mathrm{CR}_{d \leq s}(s,\theta) = \rho(\theta,s)+s$, which is increasing. The minimum in this range is $\mathrm{CR}(A_{k_\theta},\theta)$.

In all three ranges, $\mathrm{CR}(A_s,\theta) \geq \min\{\mathrm{CR}(A_0,\theta),\, \mathrm{CR}(A_{k_\theta},\theta)\}$.

\end{proof}

The transition between the two regimes is analyzed in Section~\ref{sec:critical_angle}. Specifically, $A_{k_\theta}$ is optimal for $\theta \leq \theta^*$ and $A_0$ is optimal for $\theta \geq \theta^*$, where $\theta^*$ is the critical angle determined in Theorem~\ref{thm:critical_angle}.

\subsection{The Critical Angle}\label{sec:critical_angle}

By Theorem~\ref{thm:optimal_s}, the optimal algorithm is whichever of $A_0$ and $A_{k_\theta}$ yields the smaller competitive ratio. We now determine when each is preferable.

\begin{theoremrep}\label{thm:critical_angle}
The unique $\theta^* \in (0, \pi/2)$ satisfying $\mathrm{CR}(A_0, \theta^*) = \mathrm{CR}(A_{k_{\theta^*}}, \theta^*)$ is
\begin{align*}
\theta^* &= \arccos\!\left(\frac{1}{12}\left(-8 + \sqrt[3]{1864 - 312\sqrt{33}} + 2\sqrt[3]{233 + 39\sqrt{33}}\right)\right) \\
&\approx 0.2716 \text{ rad} \approx 15.6^\circ.
\end{align*}
For $\theta \in [0, \theta^*]$, $A_{k_\theta}$ is optimal; for $\theta \in [\theta^*, \pi)$, $A_0$ is optimal.
\end{theoremrep}

\begin{proof}
When $s = 0$, $\mathrm{CR}(A_0, \theta) = \sup_{d > 0} (1+2d)/(\rho(\theta,d)+d)$. Setting $\alpha = 1$ in~\eqref{eq:cr_critical} gives the critical point equation $\rho(\theta,d) = 2 + \gamma - d(1 + 2\gamma)$.
Squaring and substituting $\rho^2 = 1 - 2d\gamma + d^2$, then dividing by $-(1+\gamma) > 0$:
\begin{align}
4\gamma\, d^2 - 4(1+\gamma)\,d + (3+\gamma) &= 0, \notag \\
d_{\mathrm{opt}} &= \frac{(1+\gamma) - \sqrt{1-\gamma}}{2\gamma}, \label{eq:d_opt_a0}
\end{align}
where the discriminant simplifies to $16(1+\gamma)^2 - 16\gamma(3+\gamma) = 16(1-\gamma) > 0$ for $\theta \in (0,\pi/2)$, and the minus branch is the valid root (the plus branch violates the positivity constraint from~\eqref{eq:cr_critical} at $s=0$).
Substituting $d_{\mathrm{opt}}$ back:
\begin{equation}\label{eq:cr_a0}
\mathrm{CR}(A_0, \theta) = \frac{3 + 2\sqrt{1-\gamma}}{(1 + \sqrt{1-\gamma})^2}, \quad \gamma = \cos\theta.
\end{equation}
This is strictly decreasing in $\theta$, since $\rho(\theta, d)$ is increasing in $\theta$ for each fixed $d > 0$, making $(1+2d)/(\rho(\theta,d)+d)$ pointwise decreasing, and hence the supremum as well.

At $\theta = 0$ ($\gamma = 1$, $\sqrt{1-\gamma} = 0$): $\mathrm{CR}(A_0, 0) = 3$ from~\eqref{eq:cr_a0}, and $k_0 = (2+\sqrt{2})/2$ gives $\mathrm{CR}(A_{k_0}, 0) = 2k_0 - 1 = 1 + \sqrt{2} < 3$, so $A_{k_\theta}$ is optimal. At $\theta = \pi/2$ ($\gamma = 0$, $\sqrt{1-\gamma} = 1$): $\mathrm{CR}(A_0, \pi/2) = 5/4$ from~\eqref{eq:cr_a0}, while $k_{\pi/2} = 1/2$ gives $\mathrm{CR}(A_{k_{\pi/2}}, \pi/2) = (1+\sqrt{5})/2 > 5/4$, so $A_0$ is optimal.
Since $\mathrm{CR}(A_0,\theta)$ and $\mathrm{CR}(A_{k_\theta},\theta)$ are continuous in $\theta$ and their ordering reverses between $\theta = 0$ and $\theta = \pi/2$, they must agree at some $\theta^* \in (0, \pi/2)$.

To obtain the closed form, substitute~\eqref{eq:cr_a0} and~\eqref{eq:k_theta_def} into $\mathrm{CR}(A_0, \theta) = \rho(\theta, k_\theta) + k_\theta$. This equation contains three nested square roots ($\sqrt{1-\gamma}$ from $\mathrm{CR}(A_0, \theta)$, the radical inside $k_\theta$, and $\rho(\theta, k_\theta)$) which we remove by repeated squaring and cross-multiplication. The result is a polynomial equation in $\gamma = \cos\theta$ of degree~$11$, which computer algebra factors as
\begin{equation}\label{eq:elim_poly}
(1+\gamma)\,(-1 + 3\gamma + \gamma^2 + \gamma^3)\,(4\gamma^3 + 8\gamma^2 - 11)\,(11 + 3\gamma - 15\gamma^2 + \gamma^3 + 4\gamma^4) = 0.
\end{equation}

Each squaring step can introduce spurious roots, so each factor must be checked against the original (unsquared) equation. The linear factor $(1+\gamma)$ has root $\gamma = -1$, corresponding to $\theta = \pi$ (the agent sits on the negative $x$-axis): direct evaluation gives $\mathrm{CR}(A_0, \pi) = 1$ but $\mathrm{CR}(A_{k_\pi}, \pi) = 3$, so this root is an artifact of squaring rather than a true solution. The degree-$4$ factor is strictly positive on $[0,1]$ and contributes no value of $\gamma = \cos\theta$ in our range. Discarding these two leaves a degree-$6$ equation that splits further into two cubics. The first, $-1 + 3\gamma + \gamma^2 + \gamma^3$, has a single real root $\gamma \approx 0.296$ in $(0,1)$; but substituting it into~\eqref{eq:cr_a0} and~\eqref{eq:k_theta_def} gives unequal competitive ratios ($1.38$ and $1.84$), so this root is also spurious. The only factor that yields a genuine solution is the second cubic,
\begin{equation}\label{eq:critical_cubic}
4\gamma^3 + 8\gamma^2 - 11 = 0,
\end{equation}
which has a unique real root
\begin{equation}\label{eq:gamma_star}
\gamma^* \;=\; \frac{1}{12}\!\left(-8 + \sqrt[3]{1864 - 312\sqrt{33}} + 2\sqrt[3]{233 + 39\sqrt{33}}\right)
\end{equation}
in $(0,1)$, and substitution confirms $\mathrm{CR}(A_0, \theta^*) = \mathrm{CR}(A_{k_{\theta^*}}, \theta^*) \approx 2.383$. Hence $\theta^* = \arccos \gamma^* \approx 0.2716$ rad $\approx 15.6^\circ$. A \texttt{Solve} call in the accompanying Mathematica script performs the elimination and the substitution checks.\footnote{\url{https://github.com/jaredraycoleman/sar_scripts}}

\end{proof}

At the boundary cases the competitive ratio recovers known results: at $\theta = 0$ the problem reduces to search and rescue on a ray~\cite{coleman_2023_sar} with optimal checkpoint $k_0 = (2+\sqrt{2})/2$ and $\mathrm{CR}(A_{k_0}, 0) = 1 + \sqrt{2} \approx 2.414$; at $\theta = \pi/2$, $\mathrm{CR}(A_0, \pi/2) = 5/4$; and as $\theta \to \pi$, $\mathrm{CR}(A_0, \theta) \to 1$.

The closed-form derivation above locates the crossover $\theta^*$ within $(0,\pi/2)$, where both candidate ratios are nontrivial. It remains to confirm that $A_0$ stays optimal on $[\pi/2,\pi)$. Observe that $\mathrm{CR}(A_0,\theta)$ is decreasing in $\theta$ (shown above), so $\mathrm{CR}(A_0,\theta) \le \mathrm{CR}(A_0,\pi/2) = 5/4$ for $\theta \ge \pi/2$, whereas the competing checkpoint algorithm satisfies $\mathrm{CR}(A_{k_\theta},\theta) \ge \mathrm{CR}(A_{k_{\pi/2}},\pi/2) = (1+\sqrt5)/2 \approx 1.618$ (the checkpoint ratio increases from $\pi/2$ toward $\mathrm{CR}(A_{k_\theta},\theta)\to 3$ as $\theta\to\pi$). Hence $A_0$ is strictly optimal throughout $[\pi/2,\pi)$, and together with Theorem~\ref{thm:optimal_s} the characterization holds for all $\theta \in [0,\pi)$.

Theorem~\ref{thm:critical_angle}, combined with Lemma~\ref{lem:reduction}, immediately yields Theorem~\ref{thm:main_general} for an arbitrary starting point $P_{\theta,r} = (r\cos\theta, r\sin\theta)$: the competitive ratio is independent of $r$, and the optimal checkpoint in the original (unscaled) coordinates is at distance $r \cdot k_\theta$ from the origin when $\theta \leq \theta^*$, or at the origin when $\theta \geq \theta^*$.

\FloatBarrier

Figure~\ref{fig:heatmaps} visualizes the results across the plane. For each starting position $P_{\theta,r} = (x, y) = (r\cos\theta, r\sin\theta)$ we evaluate the optimal algorithm of Theorem~\ref{thm:critical_angle} and plot two quantities: panel~(a) shows the competitive ratio $\mathrm{CR}(A_{s^*}, \theta)$ and panel~(b) shows the optimal checkpoint distance $s^* = r \cdot k_{|\theta|}$ (which is $0$ whenever $|\theta| \geq \theta^*$). The dashed white lines mark the critical-angle cone $|\theta| = \theta^*$.

\begin{figure}[t]
\centering
\includegraphics[width=\linewidth]{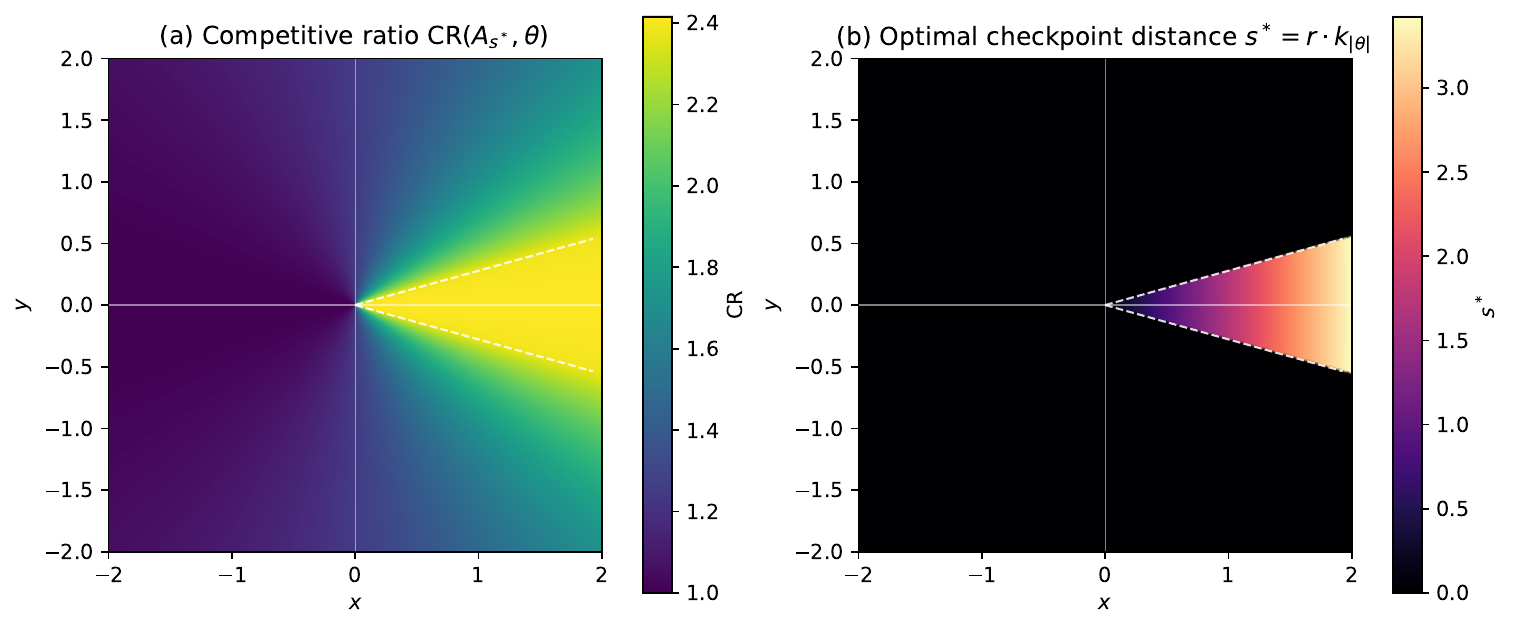}
\caption{Optimal algorithm on the plane, with the object located on the positive $x$-axis and the destination at the origin. (a)~Competitive ratio $\mathrm{CR}(A_{s^*}, \theta)$ as a function of the starting position $(x, y)$; by Lemma~\ref{lem:reduction} this depends only on $\theta = \arctan(y/x)$. (b)~Optimal checkpoint distance $s^* = r \cdot k_{|\theta|}$; outside the cone $|\theta| < \theta^*$ (dashed) the optimum is $s^* = 0$.}
\label{fig:heatmaps}
\end{figure}

Panel~(a) is radially symmetric, confirming scale-invariance (Lemma~\ref{lem:reduction}): the competitive ratio depends only on the angle $\theta$, never on the distance $r$. Within the critical cone $|\theta| < \theta^*$ the ratio rises toward $1 + \sqrt{2}$ as $|\theta| \to 0$, which recovers the one-dimensional search-and-rescue bound~\cite{coleman_2023_sar}; outside the cone it decays monotonically to $1$ as $|\theta| \to \pi$ (the agent lies behind the origin and must therefore go through it to find the target). Panel~(b) makes the phase transition at $\theta^*$ sharply visible: the optimal strategy uses a nonzero checkpoint only inside the forward cone. Together, the two panels are a direct numerical validation of Theorem~\ref{thm:critical_angle}: the discontinuity in $s^*$ at $|\theta| = \theta^*$ in panel~(b) coincides exactly with the angle at which the two candidate algorithms $A_0$ and $A_{k_\theta}$ attain equal competitive ratio in panel~(a).

\FloatBarrier
\section{Conclusion}

We presented results for a planar variant of the search and rescue problem in which an agent starting at an arbitrary position in the plane must locate an object on the positive $x$-axis and deliver it to the origin. We showed that optimal algorithms have a simple canonical form~$A_s$, derived the competitive ratio in closed form, and identified a critical angle $\theta^* \approx 15.6^\circ$ that governs the optimal strategy: for $\theta < \theta^*$ the agent benefits from an initial checkpoint on the $x$-axis, while for $\theta \geq \theta^*$ it is best to proceed directly to the origin.

Natural extensions include multiple agents, unknown starting positions, and other search domains (for example, the bounded segment). Another direction is to consider agents with limited fuel or turning costs, where the geometric trade-offs studied here would interact with additional resource constraints.

\begin{credits}
\subsubsection{\ackname}
Research of Evangelos Kranakis was supported in part by an NSERC (Natural Sciences and Engineering Research Council) of Canada grant.
\end{credits}

\bibliographystyle{splncs04}
\bibliography{refs}

\end{document}